\documentclass[draftclsnofoot,onecolumn]{IEEEtran}
\IEEEoverridecommandlockouts

\usepackage{amssymb,amsmath,amsthm,latexsym,amscd,amsfonts}
\usepackage{authblk}
\usepackage{graphics}
\usepackage{cite}
\usepackage{epsfig}
\usepackage{epstopdf}
\usepackage{bbm}
\usepackage{dsfont}
\usepackage{setspace}
\usepackage{float}
\usepackage{relsize}
\usepackage{tikz}
\usetikzlibrary{%
  arrows,%
  shapes.misc,
  shapes.arrows,%
  chains,%
  matrix,%
  positioning,
  scopes,%
  decorations.pathmorphing,
  shapes.geometric,
  shadows,
  decorations.pathreplacing,
  patterns
}

\newtheorem{mydef}{Definition}
\newtheorem{theorem}{Theorem}

\newtheorem{lemma}{Lemma}

\newtheorem{corollary}{Corollary}
\newtheorem{remark}{Remark}

\def\VR{\kern-\arraycolsep\strut\vrule &\kern-\arraycolsep}
\def\vr{\kern-\arraycolsep & \kern-\arraycolsep}
\newcommand*{\Scale}[2][4]{\scalebox{#1}{$#2$}}

\ifCLASSINFOpdf
\else
\fi
\begin{document}

\title{Perfectly Secure Index Coding}

\author[]{Mohammad Mahdi Mojahedian}
\author[]{Mohammad Reza Aref}
\author[]{Amin Gohari\thanks{This work was partially supported by Iran National Science Foundation (INSF) under contract No. 92/32575. This paper was presented in part at ISIT 2015.}}

\affil[]{\footnotesize Information Systems and Security Lab. (ISSL), Sharif University of Technology, Tehran, Iran
\\m\_mojahedian@ee.sharif.edu, aref@sharif.edu, aminzadeh@sharif.edu}

\allowdisplaybreaks
\renewcommand\Authands{ and }

\maketitle

\begin{abstract}
In this paper, we investigate the index coding problem in the presence of an eavesdropper. Messages are to be sent from one transmitter to a number of legitimate receivers who have side information about the messages, and share a set of secret keys with the transmitter. We assume perfect secrecy, meaning that the eavesdropper should not be able to retrieve any information about the message set. We study the minimum key lengths for zero-error and perfectly secure index coding problem. On one hand, this problem is a generalization of the index coding problem (and thus a difficult one). On the other hand, it is a generalization of the Shannon's cipher system. We show that a generalization of Shannon's one-time pad strategy is optimal up to a multiplicative constant, meaning that it obtains the entire boundary of the cone formed by looking at the secure rate region from the origin. Finally, we consider relaxation of the perfect secrecy and zero-error constraints to weak secrecy and asymptotically vanishing probability of error, and provide a secure version of the result, obtained by Langberg and Effros, on the equivalence of zero-error and  $\epsilon$-error regions in the conventional index coding problem.
\end{abstract}


%
\IEEEpeerreviewmaketitle

\section{Introduction}
\label{sec:intro}

An index coding problem comprises of a server, $u$ clients and a set of distinct messages $\boldsymbol{M}=\{M_1,M_2,\cdots,M_t\}$. Each client has a subset of $\boldsymbol{M}$ as its side information, and wants to learn another subset of the message set which it has not. The goal is to find the minimum number of information bits that should be broadcast by the server so that each client can recover its desired messages with \emph{zero-error} probability. This minimum required bits of information is called the optimal index code length. The index coding problem was originally introduced by Birk and Kol \cite{biko98} in a satellite communication scenario. Consider a satellite that broadcasts a set of messages to a number of clients. Each receiver may miss some of the messages due to limited storage capacity, lack of interest, interrupted reception, or any other reason. The clients then  inform the server about the messages they desire but are missing, as well as 
their side information via a feedback channel, and the server attempts to deliver their requested information by broadcasting information to all the clients. Index coding studies the efficient way of satisfying the needs of clients with minimum transmission from the satellite. To illustrate the significance of index coding, consider a communication scenario with one server, two clients and a message set $\{M_1,M_2\}$ of binary random variables. The first client has $M_2$ as side information and wants $M_1$, yet the second one has $M_1$ and wants $M_2$. The server can send the XOR of $M_1$ and $M_2$, instead of broadcasting each of them individually.

An index coding problem, in its most general case, can be represented by a directed bipartite graph \cite{nete12} or a hypergraph \cite{AlonLub2008}. However, it admits a simple graphical representation on a directed graph if each message is desired by only one client. In this case, without loss of generality one can assume that the number of  receivers and messages are the same (a client that desires two different messages can be replaced with two identical clients that desire a message each). Many of the known results in the literature are for this special case, which we also adopt in this paper.

Several upper and lower bounds are known for the optimal index code length $\ell^{\ast}(G)$ \cite{lust09,AlonLub2008,bazi11,tesa12,blkl13,blkl10,biko98,arbban13,shka14,nete12}.  Most of proposed bounds are graph-theoretic based, but \cite{arbban13} considers this problem from an information-theoretic viewpoint and computes the capacity region of index coding problem with up to five messages.  When we restrict ourselves to linear operations, the optimal linear index code is equal to a graph parameter called min-rank \cite{bazi06,bazi11}.  However, the  computation of min-rank is NP-hard \cite{peeters1996orthogonal}. Furthermore, linear index coding can be suboptimal in general \cite{lust09}. Index coding is a special case of the network coding problem. On the other hand, \cite{elrspr10,effelr12} show that any network coding problem can be reduced to an index coding problem.

Security aspects of network coding has been studied in \cite{bhattad2005weakly,bloch2011physical,jaggi2007resilient,yeung2008information}. In particular, secure throughput of a network coding problem in the presence of an active adversary who can eavesdrop and corrupt some links are studied. A similar problem with active adversaries has been studied in \cite{dau2011secure} for the linear index coding problem.

In this paper, we study secrecy in index coding from a different perspective. Our approach is similar to that of Shannon  in his seminal paper \cite{shannon49}. He analyzed the cipher system shown in Fig.~\ref{fig:Shannon_Cipher}, comprising of a message $M$, a cipher text $C$, and a key $K$ - a secret common randomness shared between the sender and the legitimate receiver. The sender wishes to transmit $M$ to the legitimate receiver while keeping it secret from the eavesdropper. To this end, the sender transmits $C$ (a function of $M$ and $K$) on a public noiseless channel. By receiving $C$, the eavesdropper should not be able to attain any information about $M$. Shannon adopted the notion of \emph{perfect secrecy}, of statistical independence between the message and the cipher text, \emph{i.e.,} $I(M;C)=0$. Moreover, Shannon assumed \emph{zero-error} recovery of the message: the legitimate receiver should be able to retrieve the message from $C$ and $K$, imposing the constraint $H(M|K,C)=0$. Shannon proved that the cipher system of Fig.~\ref{fig:Shannon_Cipher} is perfectly secure, if the following inequality is satisfied:
\begin{equation}
\label{Shannon_Condition}
H(K)\geq H(M).
\end{equation} 
Roughly speaking, perfect secrecy is possible if and only if the key length is  greater than or equal to the message length. Achievability follows from the one-time pad scheme.
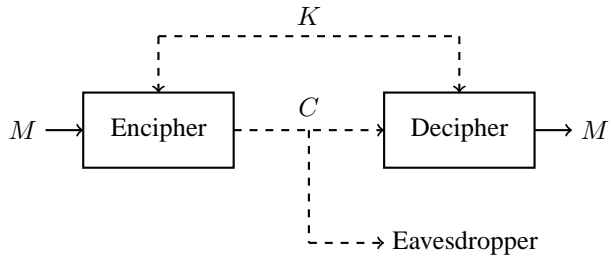
\begin{figure}[ht]
\centering
\begin{tikzpicture}

\draw[thick,->] (-.5,.5) -- (0,.5) node [black,inner sep=10pt, pos=0.5, anchor=east] {$M$};

\draw[thick] (0,0) rectangle (2,1) node [black,inner sep=0pt, pos=0.5, anchor=center] {$\text{Encipher}$};

\draw[thick,dashed,->] (2,.5) -- (4,.5) node [black,inner sep=5pt, pos=0.5, anchor=south] {$C$};

\draw[thick,dashed] (3,.5) -- (3,-1);
\draw[thick,dashed,->] (3,-1) -- (4,-1) node [black,inner sep=17pt, pos=0.5, anchor=west] {$\text{Eavesdropper}$};

\draw[thick] (4,0) rectangle (6,1) node [black,inner sep=0pt, pos=0.5, anchor=center] {$\text{Decipher}$};

\draw[thick,->] (6,.5) -- (6.5,.5) node [black,inner sep=10pt, pos=0.5, anchor=west] {$M$};

\draw[thick,dashed,->] (1,1.75) -- (1,1);
\draw[thick,dashed] (1,1.75) -- (5,1.75) node [black,inner sep=4pt, pos=0.5, anchor=south] {$K$};
\draw[thick,dashed,->] (5,1.75) -- (5,1);

\end{tikzpicture}
\caption{Shannon cipher system.} 
\label{fig:Shannon_Cipher}
\end{figure}

The goal of this paper is to derive a condition similar to inequality \eqref{Shannon_Condition} for a general \emph{zero-error and perfectly secure} index coding problem (observe that Shannon's cipher system is a special index coding problem with one receiver). Consider a scenario with $t$ legitimate receivers, an eavesdropper, and a set of keys $\boldsymbol{K}$ shared between the sender and the legitimate receivers. The question is to find the minimum entropy of keys required for perfect secrecy. Moreover, the effect of perfect secrecy condition on the optimal index code length is studied.

This paper deals with the three main theorems. The first one, proves a relation between secure and conventional (without secrecy) index coding problems. For a secure index coding problem, we propose a generalized one-time pad strategy which is shown to be optimal up to a multiplicative constant. The second theorem is a linear version of the first theorem, and finally, the last theorem discusses the equivalency of rate region in weakly and perfectly secure index coding problems (with zero or vanishing error probabilities).

The rest of this paper is organized as follows. In Section \ref{Def_Note}, the system model is defined. Section \ref{Main_Results_Sec} lays out the main results. We state the proofs in Section \ref{sec:proofs}. Section \ref{Conclusion_Sec} concludes this paper.

\textbf{Notation.} Random variables are shown in capital letters, whereas their realizations are shown in lowercase letters. Bold letters are used to denote sets or vectors. Alphabet set of random variables are shown in calligraphic font.  We use $[t]$ to denote $\{1,2,\cdots, t\}$ and $X_\mathbf{S}$ for some subset $\mathbf{S}$ of indices to denote the collection of $(X_s: s\in \mathbf{S})$. We use $[a]_{+}$ to denote $a$ if it is non-negative and zero otherwise. We use the term ``conventional index code'' to denote a classical index coding problem with no adversary and secret keys.

\section{System Model}
\label{Def_Note}
Conventional index coding is the problem of sending a set of $t$ messages $\boldsymbol{M}=\{M_1,M_2,\cdots,M_t\}$ to $t$ receivers. The $i$-th receiver wants the message $M_i$, having a subset of remaining messages $\boldsymbol{M}\setminus M_i=\{M_1,M_2,\cdots,M_{i-1},\\M_{i+1},\cdots,M_t\}$ as side information. The side information set of $i$-th receiver is shown by $\boldsymbol{S}_i$. The goal is to minimize the amount of information that should be broadcast to the receivers for decoding their desired messages without any error. 

Now, assume that an eavesdropper coexists with the legitimate receivers. Just like legitimate receivers, the eavesdropper receives the index code $C$. However, we require that  the eavesdropper should not be able to obtain any information about message set $\boldsymbol{M}$ from index code $C$ (perfect secrecy). From an information theoretic perspective, the mutual information of $\boldsymbol{M}$ and $C$ should be zero. To accomplish this, we assume that the transmitter and the legitimate receivers share common and private secret keys. The common key $K$ is shared among the sender and all of the legitimate receivers, and the private key $K_i, i\in[t]$ is shared between the sender and the $i$-th receiver. We are interested in the minimum entropy of the keys needed for perfect secrecy. 

Below, we formally define a secure index code.

\begin{mydef}[Secure Index Code]
\label{Def_Secure_IC}
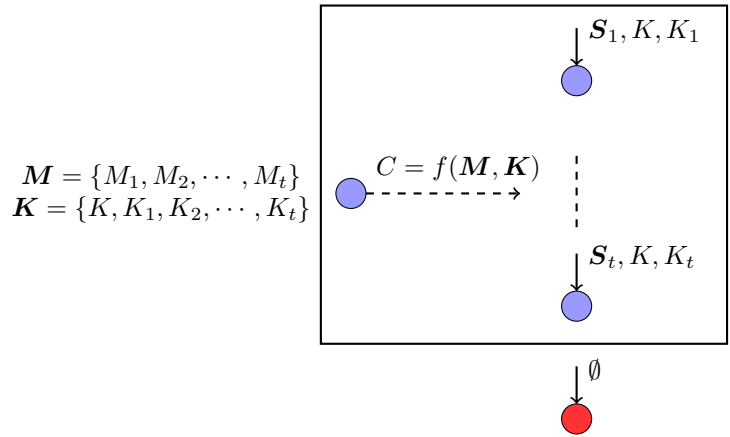
\begin{figure}[!b]
\centering
\begin{tikzpicture}

\fill[blue!40!white, draw=black] (0,0) circle (.2cm) node[black,inner sep=15pt, pos=0, anchor=east,align=center] {$\boldsymbol{M}=\{M_1,M_2,\cdots,M_t\}$\\$\boldsymbol{K}=\{K,K_1,K_2,\cdots,K_t\}$};
\fill[blue!40!white, draw=black] (3,1.5) circle (.2cm);
\fill[blue!40!white, draw=black] (3,-1.5) circle (.2cm);
\fill[red!80!white, draw=black] (3,-3) circle (.2cm);

\draw[thick,->] (3,2.2) -- (3,1.7) node [black,inner sep=4pt, pos=0.1, anchor=west] {$\boldsymbol{S}_1,K,K_1$};
\draw[thick,->] (3,-.8) -- (3,-1.3) node [black,inner sep=4pt, pos=0.1, anchor=west] {$\boldsymbol{S}_t,K,K_t$};
\draw[thick,->] (3,-2.3) -- (3,-2.8) node [black,inner sep=4pt, pos=0.1, anchor=west] {$\emptyset$};

\draw[thick,dashed,->] (.2,0) -- (2.25,0) node [black,inner sep=4pt, pos=.6, anchor=south] {$C=f(\boldsymbol{M},\boldsymbol{K})$};

\draw[thick,dashed] (3,.5) -- (3,-.5);

\draw[thick] (-.4,-2) rectangle (5,2.5);

\end{tikzpicture}
\caption{The schematic of secure index coding scenario.} 
\label{fig:Secure_IC}
\end{figure}
Consider the scenario of Fig.~\ref{fig:Secure_IC} consisting of a sender (who broadcasts data), $t$ legitimate receivers, and an illegal receiver named as the eavesdropper. Also, assume a key set $\boldsymbol{K}=\{K,K_1,K_2,\cdots,K_t\}$ of common and private keys. A secure index coding scheme consists of an encoder and $t$ decoders satisfying the perfect secrecy condition, defined as follows:

1- Encoder: An encoder $f$ maps the message set $\boldsymbol{M}$ and the key set $\boldsymbol{K}$ to a code symbol $C\in \mathcal{C}$,
\begin{equation*}
f: \mathcal{M}_1\times\mathcal{M}_2\times\cdots\times\mathcal{M}_t
\times\mathcal{K}\times\mathcal{K}_1\times\cdots\times\mathcal{K}_t\times \mathcal W
\rightarrow \mathcal{C}.
\end{equation*}
where $\mathcal{M}_i$, $\mathcal{K}$, $\mathcal{K}_i$, and $\mathcal{C}$ are the alphabet sets of $M_i$, $K$, $K_i$, and $C$, respectively. Here $\mathcal{W}$ is the alphabet set for $W$, which is the private source of randomness for the encoder, independent of all previously defined random variables. If $|\mathcal W|=1$, the encoder will be deterministic. Random variable $W$ is known only to the encoder.

2- Decoder: A decoder $g_i, i=1,\cdots,t$ recovers $M_i$ from code symbol $C$, its side information $\boldsymbol{S}_i$, as well as the keys $K$ and $K_i$,
\begin{equation}
g_i: \mathcal{C}\times \mathcal{S}_i\times\mathcal{K}\times\mathcal{K}_i\rightarrow \mathcal{M}_i.\label{eqn:defg-i1}
\end{equation} 
The recovery is exact: $g_i(c, \textbf{s}_i, k, k_i)=m_i$. Thus, for any $i$ and arbitrary input distribution on the message set $\boldsymbol{M}$, we should have:
\begin{equation*}
H(M_i|C,\boldsymbol{S}_i,K,K_i)=0.
\end{equation*}
It means that each receiver should be able to retrieve its desired message from its side information, the code $C$, as well as the keys $K$ and $K_i$ with error probability zero. 

3- Perfect secrecy condition: assuming that $K$ and $K_i$ are mutually independent and uniform over their alphabet sets, the conditional pmf
$p(C=c|\boldsymbol{M}=\boldsymbol{m})$ should not depend on the value of $\boldsymbol{m}$, for any given $c$. Equivalently, for any distribution on input message $\boldsymbol{M}$, we should have:
\begin{equation}
I(\boldsymbol{M};C)=0, \qquad \forall p_{\boldsymbol{M}}(\boldsymbol{m}) \label{eqn:addedM}
\end{equation} 
as long as the message set $\boldsymbol{M}$, the key set $\boldsymbol{K}$ and private randomness $W$ are mutually independent.

4- Rate vector: corresponding to a secure index code, a rate vector \begin{align}\boldsymbol{r}=(r_1,r_2,\cdots,r_t, r_k, r_{k_1}, \cdots, r_{k_t})\label{def:r-vec}\end{align} is defined, where 
\begin{align*}
r_i=\frac{\log\lvert\mathcal{M}_i\rvert}{\log\lvert\mathcal{C}\rvert},\qquad
r_k=\frac{\log\lvert\mathcal{K}\rvert}{\log\lvert\mathcal{C}\rvert},\qquad
r_{k_i}=\frac{\log\lvert\mathcal{K}_i\rvert}{\log\lvert\mathcal{C}\rvert}.
\end{align*}
\end{mydef}

\begin{remark}\emph{
Throughout, we reserve the notation $``r_k"$ for the rate of common key. It should not be confused with $r_1, r_2, \cdots, r_t$ which are message rates. When we write $r_i$ for a variable $i\in[t]$, we mean one of  $r_1, r_2, \cdots, r_t$, and not $r_k$.} \end{remark}

\begin{remark}\emph{ A secure index code is an extension of the conventional index code with no adversary.  If we consider a zero-error index code that does not necessarily satisfy the perfect secrecy constraint, and has a  rate vector of the following form,
 \begin{align}\boldsymbol{r}=(r_1,r_2,\cdots,r_t, 0, 0, \cdots, 0)\label{def:r-vec1},\end{align}
\emph{i.e.,} no secret keys exist $r_k=r_{k_i}=0$, then we get a conventional zero-error index code with rate vector
 \begin{align}(r_1, r_2, \cdots, r_t)\label{def:r-vec12}.\end{align}}
\end{remark}

Linear index codes form a subclass of the general problem, in which both encoder and decoders are linear functions.

\begin{mydef}[Linear Index Code]\label{def2}
A linear index code includes a linear encoder and $t$ linear decoders so that:

1- Encoder: A \underline{linear} function $f$ mapping the message set $\boldsymbol{M}$ and secret keys $\boldsymbol{K}$ to a code symbol $C\in\mathbb{F}^{l}$,
\begin{equation*}
f: \mathbb{F}^{l_1}\times\mathbb{F}^{l_2}\times\cdots\times\mathbb{F}^{l_t}\times\mathbb{F}^{l_k}\times\mathbb{F}^{l_{k_1}}
\times\mathbb{F}^{l_{k_2}}\times\cdots\times\mathbb{F}^{l_{k_t}}\times \mathbb{F}^{l_w}\rightarrow \mathbb{F}^{l}.
\end{equation*}
where $\mathbb{F}$ is a finite field, $l_i$, $l_k$, $l_{k_i}$, $l_w$ and $l$ are respectively the length of message $M_i$, the length of the common key $K$, the length of private key $K_i$, the length of private randomness $W$,  and the length of index code $C$. In other words, $M_i$, $K$, $K_i$, $W$ and $C$ are sequences of length $l_i$, $l_k$, $l_{k_i}$, $l_w$ and $l$ in the field $\mathbb{F}$. 

2- Decoder: A linear function $g_i$ for $i\in[t]$ that acts on code symbol $C$, side information $\boldsymbol{S}_i$ and secret keys $K, K_i$ to  recover the message $M_i$
\begin{equation*}
g_i: \mathbb{F}^{l}\times \mathcal{S}_i\times\mathbb{F}^{l_k}\times\mathbb{F}^{l_{k_i}}\rightarrow \mathbb{F}^{l_i}.
\end{equation*} 

3- Rate vector: the rate vector of linear index coding is defined as follows:
$$\boldsymbol{r}=(r_1,r_2,\cdots,r_t, r_k, r_{k_1}, \cdots, r_{k_t})$$where
\begin{align*}
r_i=\frac{l_i}{l},\qquad
r_k=\frac{l_{k}}{l},\qquad
r_{k_i}=\frac{l_{k_i}}{l}.
\end{align*}
\end{mydef}

Each code symbol is a linear function of the components of $M_{i}$, $K$ and $K_{i}$, \emph{i.e.,}
$$C_i=\sum_{p=1}^{l_k}\alpha^{i}_{p}K(p)+\sum_{j=1}^{t}\sum_{p=1}^{l_{k_j}}\beta^{i}_{jp}K_{j}(p)+\sum_{j=1}^t\sum_{p=1}^{l_j}\gamma^{i}_{jp}M_{j}(p)+\sum_{p=1}^{l_w}\psi_{p}^iW(p)$$
for some coefficients $\alpha^{i}_{p}$, $\beta^{i}_{jp}$, $\gamma^{i}_{jp}$ and $\psi_{p}^i$ in $\mathbb{F}$. Here,
\begin{align*}
M_i&=(M_i(1), M_i(2), \cdots, M_i(l_i)),\\
K&=(K(1), K(2), \cdots, K(l_k)),\\
K_i&=(K_i(1), K_i(2), \cdots, K_i(l_{k_i})),
\end{align*}
and
\begin{align*}
W&=(W(1), W(2), \cdots, W(l_w))
\end{align*}
are strings of symbols in $\mathbb{F}$.
Thus, the encoding scheme in linear index coding problem has the following matrix representation\footnotesize{
\begin{align}
C=
\begin{pmatrix}
C_1\\
C_2\\
\vdots\\
C_l
\end{pmatrix}=
\begin{pmatrix}
\boldsymbol{\alpha}^1&\boldsymbol{\beta}_1^1&\cdots &\boldsymbol{\beta}_{t}^1&\boldsymbol{\psi}^1
&\boldsymbol{\gamma}_1^1&\cdots &\boldsymbol{\gamma}_{t}^1\\
\boldsymbol{\alpha}^2&\boldsymbol{\beta}_1^2&\cdots &\boldsymbol{\beta}_{t}^2&\boldsymbol{\psi}^2
&\boldsymbol{\gamma}_1^2&\cdots &\boldsymbol{\gamma}_{t}^2\\
\vdots &\vdots&&\vdots&\vdots&\vdots&&\vdots\\
\boldsymbol{\alpha}^l&\boldsymbol{\beta}_1^l&\cdots &\boldsymbol{\beta}_{t}^l&\boldsymbol{\psi}^l
&\boldsymbol{\gamma}_1^l&\cdots &\boldsymbol{\gamma}_{t}^l
\end{pmatrix}
\begin{pmatrix}
K\\
K_1\\
\vdots\\
K_t\\
W\\
M_1\\
\vdots\\
M_{t}
\end{pmatrix},\label{eqn5}
\end{align}}\normalsize
where
\begin{equation*}
\begin{matrix}
\boldsymbol{\alpha}^i&= &(\alpha_{1}^i&\alpha_{2}^i&\cdots &\alpha_{l_k}^i),\\
\boldsymbol{\beta}_{j}^i&=& (\beta_{j1}^i&\beta_{j2}^i&\cdots &\beta_{jl_{k_j}}^i\hspace{-.3mm}),\\
\boldsymbol{\gamma}_{j}^i&= &(\gamma_{j1}^i&\gamma_{j2}^i&\cdots &\gamma_{jl_j}^i),\\
\boldsymbol{\psi}^i&= &(\psi_{1}^i&\psi_{2}^i&\cdots &\psi_{l_w}^i).
\end{matrix}
\end{equation*}
which construct the code generation matrix shown by $\Pi$ throughout this paper.

\begin{mydef}[One-Shot and Asymptotic Index Coding]
In the one-shot case, a single use of the index coding problem is considered. In other words, there are fixed message alphabet sets $\mathcal{M}_1,\mathcal{M}_2,\cdots,\mathcal{M}_t$, and the goal is to find an index code with minimum amount of keys and public communication that would ensure zero-error perfect secrecy. In other words, we are looking for the set of all possible minimal rate vectors
$$\boldsymbol{r}=(r_1,r_2,\cdots,r_t, r_k, r_{k_1}, \cdots, r_{k_t}),$$
as in \eqref{def:r-vec} for fixed alphabet sets $\mathcal{M}_1,\mathcal{M}_2,\cdots,\mathcal{M}_t$.

On the other hand, the asymptotic case asks for the set of all possible rate vectors $\boldsymbol{r}$ that are asymptotically achievable, \emph{i.e.,} there exists a sequence of zero-error and perfectly secure index codes whose rate vectors converge to $\boldsymbol{r}$.
\end{mydef}

\begin{mydef}
The asymptotic secure index coding region, $\mathcal{R}_{\mathsf{Secure}}$, is defined to be the set of all asymptotically achievable tuples
$$\boldsymbol{r}=(r_1,r_2,\cdots,r_t, r_k, r_{k_1}, \cdots, r_{k_t}).$$
The conventional asymptotic index coding region is defined similarly using the achievable rate vectors as in equation \eqref{def:r-vec12}. We denote this regions by $\mathcal{R}$.
\end{mydef}
\begin{remark}\emph{
 Observe that the region $\mathcal{R}_{\mathsf{Secure}}$ specifies $\mathcal{R}$ since
 \begin{align}\boldsymbol{r}=(r_1,r_2,\cdots,r_t, \infty, \infty, \cdots, \infty),\end{align}
is in the secure rate region if and only if $(r_1,r_2,\cdots,r_t)$ is in the conventional zero-error index code. Thus, finding the region $\mathcal{R}_{\mathsf{Secure}}$ is at least as difficult as finding $\mathcal{R}$. We will show that finding the difficulty of finding $\mathcal{R}_{\mathsf{Secure}}$ when viewed from the origin is as difficult as finding $\mathcal{R}$.}
\end{remark}

\begin{remark}\emph{
In spite of the fact that the asymptotic case is commonly related to vanishing instead of zero probability of error, it has been shown in \cite{laneff11} that in the conventional index coding (with no adversary or secret keys), zero and asymptotic error capacities are the same. }
\end{remark}
\begin{remark}\emph{
Clearly, were a rate vector $\boldsymbol{r}$ one-shot achievable, it is also asymptotically achievable. Also, if $(r_1,r_2,\cdots,r_t, r_k, r_{k_1}, \cdots, r_{k_t})$ is achievable, then so is $(r_1-\alpha_1,r_2-\alpha_2,\cdots,r_t-\alpha_t, r_k+\beta_{k}, r_{k_1}+\beta_{k_1}, \cdots, r_{k_t}+\beta_{k_t})$  for any non-negative values of $\alpha_i$ and $\beta_k$ and $\beta_{k_i}$.
}
\end{remark}


\section{Main Results}
\label{Main_Results_Sec}

\subsection{Generalized One-Time Pad Strategy}

Without loss of generality, let us assume a three-user case. As shown in Fig. \ref{fig:Generalized_OTP}, a possible strategy for the secure index coding problem is to use private key $K_i$ and XOR it with part of the message $M_i$. This way, we can privately communicate  parts of the messages. Then, for the remaining parts of the messages, we can find the optimal index code and XOR it with the common key $K$. This can be seen as a generalized version of one-time pad scheme which is used in the Shannon's cipher system. We will prove that this modified version of one-time pad strategy is optimal up to a multiplicative constant.

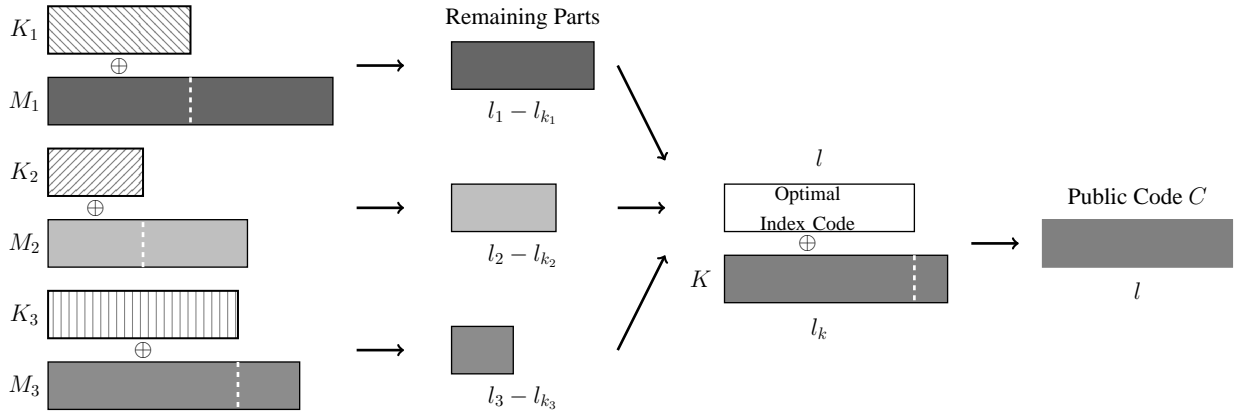
\begin{figure}[H]
\centering
\resizebox {\columnwidth} {!} {
\begin{tikzpicture}[scale=1]

\draw[very thick,black,pattern=north west lines, pattern color=black!50] (0,0) rectangle (3,1);
\node at (-.5,.5) [align=center] {\Large $K_1$};

\node at (1.5,-.25) [align=center] {\Large $\oplus$};

\draw[thick,black,fill=black!60] (0,-1.5) rectangle (6,-.5);
\node at (-.5,-1) [align=center] {\Large $M_1$};
\draw[ultra thick,dashed,white] (3,-.5)--(3,-1.5);

\draw[very thick,black,pattern=north east lines, pattern color=black!50] (0,-3) rectangle (2,-2);
\node at (-.5,-2.5) [align=center] {\Large $K_2$};

\node at (1,-3.25) [align=center] {\Large $\oplus$};

\draw[thick,black,fill=black!25] (0,-4.5) rectangle (4.2,-3.5);
\node at (-.5,-4) [align=center] {\Large $M_2$};
\draw[ultra thick,dashed,white] (2,-4.5)--(2,-3.5);

\draw[very thick,black,pattern=vertical lines, pattern color=black!50] (0,-6) rectangle (4,-5);
\node at (-.5,-5.5) [align=center] {\Large $K_3$};

\node at (2,-6.25) [align=center] {\Large $\oplus$};

\draw[thick,black,fill=black!45] (0,-7.5) rectangle (5.3,-6.5);
\node at (-.5,-7) [align=center] {\Large $M_3$};
\draw[ultra thick,dashed,white] (4,-6.5)--(4,-7.5);


\draw[ultra thick,->] (6.5,-.25)--(7.5,-.25);
\draw[ultra thick,->] (6.5,-3.25)--(7.5,-3.25);
\draw[ultra thick,->] (6.5,-6.25)--(7.5,-6.25);

\node at (10,.75) [align=center] {\Large Remaining Parts};
\draw[thick,black,fill=black!60] (8.5,-.75) rectangle (11.5,.25);
\node at (10,-1.25) [align=center] {\Large $l_1-l_{k_1}$};

\draw[thick,black,fill=black!25] (8.5,-3.75) rectangle (10.7,-2.75);
\node at (10,-4.25) [align=center] {\Large $l_2-l_{k_2}$};

\draw[thick,black,fill=black!45] (8.5,-6.75) rectangle (9.8,-5.75);
\node at (10,-7.25) [align=center] {\Large $l_3-l_{k_3}$};


\draw[ultra thick,->] (12,-.25) -- (13,-2.25);
\draw[ultra thick,->] (12,-3.25) -- (13,-3.25);
\draw[ultra thick,->] (12,-6.25) -- (13,-4.25);


\draw[thick,black] (14.25,-3.75) rectangle (18.25,-2.75);

\node at (16,-3.25) [align=center] {\large Optimal\\\large Index Code};

\node at (16.25,-2.2) [align=center] {\Large $l$};

\node at (16,-4) [align=center] {\Large $\oplus$};

\draw[thick,black,fill=gray] (14.25,-5.25) rectangle (18.95,-4.25);
\node at (16.25,-5.8) [align=center] {\Large $l_k$};
\node at (13.75,-4.75) [align=center] {\Large $K$};
\draw[ultra thick,dashed,white] (18.25,-4.25)--(18.25,-5.25);

\draw[ultra thick,->] (19.45,-4)--(20.45,-4);

\draw[thick,gray,fill=gray] (20.95,-4.5) rectangle (24.95,-3.5);
\node at (22.95,-3) [align=center] {\Large Public Code $C$};

\node at (22.95,-5) [align=center] {\Large $l$};

\end{tikzpicture}
}
\caption{Generalized one-time pad strategy. Here message lengths, common key length, private key lengths and the index code length, are denoted by the $l_i$, $l_k$, $l_{k_i}$ and $l$, respectively.} 
\label{fig:Generalized_OTP}
\end{figure}

In  Fig.~\ref{fig:Generalized_OTP}, the remaining parts of the messages are secured by XORing them with symbols of $K$. Therefore, $l_k$ should be greater than or equal to the length of optimal index code length
needed for communicating the remaining parts of the messages, \emph{i.e.,} $l_k\geq l$. In order to be able to utilize the generalized one-time pad strategy, a further constraint needs to be met. In the index code for the remaining parts of the messages, we have compressed $l_i-l_{k_i}$ symbols from user $i$ into $l$ index symbols, and therefore the rate of user $i$ in this index code is equal to
\begin{equation*}
\frac{l_i-l_{k_i}}{l}\overset{(a)}{\geq}\frac{r_i-r_{k_i}}{r_k}\qquad i=1,2,3.
\end{equation*}
where $(a)$ comes from perfect secrecy condition. Thus, the rate vector \begin{equation}\label{eqn:jadid}\left(\frac{r_1-r_{k_1}}{r_k},\frac{r_2-r_{k_2}}{r_k},\frac{r_3-r_{k_3}}{r_k}\right),\end{equation} must belong to the conventional index coding problem rate region (without secrecy constraints).
The generalized one-time pad strategy works if the rate tuple given in equation \eqref{eqn:jadid}, corresponding to the secure index coding rate tuple $(r_1, r_2, r_3, r_k, r_{k_1}, r_{k_2}, r_{k_3})$, belongs to the conventional index coding region. The main theorem of this paper provides a converse to this result, up to a constant multiplicative factor.

\subsection{Optimality of generalized one-time pad up to a multiplicative constant}

\begin{theorem}
\label{Main_theorem} Given non-negative values for $r_1,r_2,\cdots,r_t, r_k, r_{k_1}, \cdots, r_{k_t}$, the following three statements are equivalent:
\begin{align*}(a):&\qquad \exists \alpha>0:\quad \alpha\cdot(r_1,r_2,\cdots,r_t, r_k, r_{k_1}, \cdots, r_{k_t})\in \mathcal{R}_{\mathsf{Secure}},\\
&\Longleftrightarrow\\
(b):&\qquad \exists \alpha>0:\quad\alpha\cdot([r_1-r_{k_1}]_+,[r_2-r_{k_2}]_+,\cdots, [r_t-r_{k_t}]_+, r_k, 0, \cdots, 0)\in \mathcal{R}_{\mathsf{Secure}},
\\ &\Longleftrightarrow\\
(c):&\qquad(\frac{[r_1-r_{k_1}]_+}{r_k}, \frac{[r_2-r_{k_2}]_+}{r_k}, \cdots, \frac{[r_t-r_{k_t}]_+}{r_k})\in \mathcal{R}.\end{align*}
Similarly, 
\begin{align*}(a):&\quad \exists \alpha>0:\quad \alpha\cdot(r_1,r_2,\cdots,r_t, r_k, r_{k_1}, \cdots, r_{k_t})\in \mathcal{R}_{\mathsf{Secure-Linear}},\\
&\Longleftrightarrow\\
(b):&\quad \exists \alpha>0:\quad\alpha\cdot([r_1-r_{k_1}]_+,[r_2-r_{k_2}]_+,\cdots,[r_t-r_{k_t}]_+,r_k, 0, \cdots, 0)\in \mathcal{R}_{\mathsf{Secure-Linear}},
\\ &\Longleftrightarrow\\
(c):&\quad(\frac{[r_1-r_{k_1}]_+}{r_k}, \frac{[r_2-r_{k_2}]_+}{r_k}, \cdots, \frac{[r_t-r_{k_t}]_+}{r_k})\in\mathcal{R}_{\mathsf{Linear}}.\end{align*}

Here, to disambiguate the special case $r_k=0$ showing up in the denominator, we define $c/0$ to be zero if $c=0$, and infinity otherwise.
\end{theorem}
\begin{corollary}
\label{Theorem_Private_Key}\emph{
In the case that only private keys $K_i, i\in[t]$ are available, \emph{i.e.,} $r_k=0$, perfect secrecy is possible if and only if
\begin{equation*}
r_{k_i}\geq r_i, i\in[t].
\end{equation*}
This is because if $r_{k_i}<r_i$ for some $i$, then $[r_i-r_{k_i}]_+/{r_k}$ will be infinity. This is a contradiction since the rates in index coding are at most one. \\
Clearly, $r_{k_i}\geq r_i$ implies that we can do separate one-time pad on individual messages. With this strategy, the length of public communication $l$ will be equal to $\sum_{i=1}^t l_{k_i}$. It turns out that we cannot achieve zero-error perfect security with $l<\sum_{i=1}^t l_{k_i}$ in this case.}
\end{corollary}
{
\begin{remark}
The Shannon cipher system is a special case of the secure index coding problem. In the Shannon cipher system, where we have one legitimate receiver, perfect secrecy condition necessitates $r/r_k\leq 1$, where $r$ is the message rate and $r_k$ is the key rate. Similarly, if we consider no private keys, the third statement of the above-mentioned theorem implies that $r_i/r_k\leq 1, i\in[t]$ which is an extension of the Shannon perfect secrecy condition to multiple receivers. 
\end{remark}
\begin{remark}
Consider the first and third parts of the theorem. The factor $\alpha$ in the statement $(a)$ specifies the cone of the secure rate region (if $\alpha$ multiplied by the rate vector is in the $\mathcal{R}_{\mathsf{Secure}}$, the rate vector itself belongs to the cone of this region when viewed from the origin). Hence, as shown in the Fig. \ref{fig:Cone_secure_region}, the theorem intuitively states that the conventional index coding problem rate region determines the cone of the secure rate region. Moreover, the introduced generalized one-time pad strategy gives an achievable rate region which is a subset of $\mathcal{R}_{\mathsf{Secure}}$ and has a cone being the same as that of the secure rate region.
\end{remark}

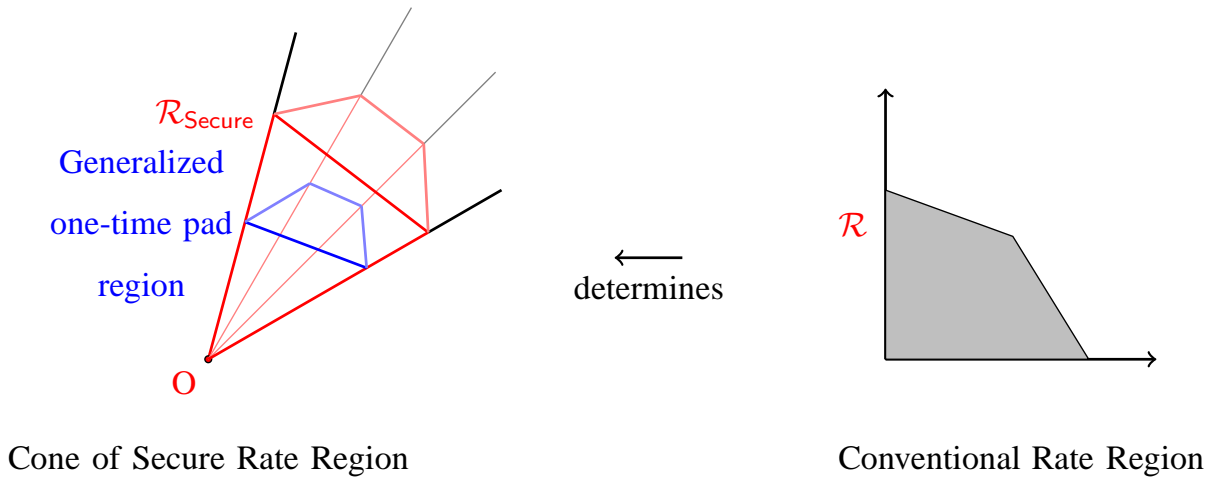
\begin{figure}[H]
\centering
\resizebox {\columnwidth} {!} {
\begin{tikzpicture}[scale=.7]

\draw[thick,red] (0,0) -- (75:3.75) node [red,inner sep=5pt, pos=1, anchor=east] {$\mathcal{R}_{\mathsf{Secure}}$};

\draw[fill=red] (0,0) circle (.05) node [red,below left] {O};
\node at (0,-1.5) [align=center] {Cone of Secure Rate Region};
\draw[red!50] (0,0) -- (60:4.5);
\draw[red!50] (45:0) -- (45:4.5);
\draw[thick,red] (0,0) -- (30:3.75);

\draw[gray] (60:4.5) -- (60:6);
\draw[thick] (75:3.75) -- (75:5);
\draw[gray] (45:4.5) -- (45:6);
\draw[thick] (30:3.75) -- (30:5);

\draw[thick,blue] (30:2.7) -- (75:2.1) node [blue,inner sep=7pt, pos=.9, anchor=east,align=center] {Generalized\\one-time pad \\region};
\draw[thick,blue!50] (75:2.1)--(60:3);
\draw[thick,blue!50] (60:3)--(45:3.2);
\draw[thick,blue!50] (45:3.2)--(30:2.7);

\draw[thick,red] (30:3.75) --(75:3.75);
\draw[thick,red!50] (75:3.75)--(60:4.5);
\draw[thick,red!50] (60:4.5)--(45:4.5);
\draw[thick,red!50] (45:4.5)--(30:3.75);

\draw[->,thick] (10,0)--(10,4) node [red,inner sep=5pt, pos=.5, anchor=east] {$\mathcal{R}$}; 
\draw[->,thick] (10,0)--(14,0);
\node at (12,-1.5) [align=center] {Conventional Rate Region};
\draw[fill=gray!50] (10,2.5)--+(-20:2)--(13,0)--(10,0)--cycle; 

\draw[->,thick] (7,1.5)--(6,1.5) node [inner sep=5pt, pos=.5, anchor=north] {determines};

\end{tikzpicture}
}
\caption{Conventional index coding region determines the cone of the secure rate region. The generalized one-time pad strategy obtains the entire boundary of the cone.} 
\label{fig:Cone_secure_region}
\end{figure}
}

Theorem \ref{Main_theorem_Linear} presents a similar statement to the Theorem \ref{Main_theorem} for the linear case.

\begin{theorem}
\label{Main_theorem_Linear} Suppose we are given message alphabet sets $\mathcal{M}_1,\mathcal{M}_2,\cdots,\mathcal{M}_t$ where $\mathcal{M}_i=\mathbb{F}^{l_i}$ for some finite field $\mathbb{F}$. Then, there exists a linear zero-error perfectly secure index code with key lengths $(l_k, l_{k_1}, \cdots, l_{k_t})$ and code length $l$, if and only if there exists a linear zero-error conventional index code (no secrecy) with code length $l_k$ for message sets $\widetilde{\mathcal{M}}_1,\widetilde{\mathcal{M}}_2,\cdots,\widetilde{\mathcal{M}}_t$ where $\widetilde{\mathcal{M}}_i=\mathbb{F}^{[l_i-l_{k_i}]_{+}}$ in which $[a]_{+}$ is $a$ if it is non-negative, and is zero otherwise.
\end{theorem}

\subsection{Variations on security and reliability constraints}
Our proof of Theorem \ref{Main_theorem} requires us to study the perfectly secure achievable rates under an asymptotically vanishing error criterion (rather than the exactly zero-error criterion). For this, we develop a result that can be understood as a perfectly secure version of the result of \cite{effelr12} on the equivalence of asymptotically zero and exactly zero network coding rates. Below, we provide a more general result than the one needed in the proof of  Theorem \ref{Main_theorem} by comparing achievable rates of weakly secure codes with an asymptotically vanishing error, with those of perfectly secure zero-error codes. To proceed, let us define two other secrecy conditions, in addition to the perfect secrecy constraint mentioned in part 3 of Definition \ref{Def_Secure_IC}.

\begin{mydef}[Strong Secrecy and Vanishing Error Probability] A rate vector \begin{align}\boldsymbol{r}=(r_1,r_2,\cdots,r_t, r_k, r_{k_1}, \cdots, r_{k_t})\end{align} is strongly secure achievable with a vanishing probability of error if for any $\epsilon>0$, there is a code whose rate vectors is in the $\epsilon$ distance of $\boldsymbol{r}$, and furthermore, assuming a uniform and independent distribution over the messages in $\boldsymbol{M}$, the error probability of the code is less than or equal to $\epsilon$ and 
\begin{equation*}
\|p_{\boldsymbol{M},C}-p_{\boldsymbol{M}}p_{C}\|_1\leq\epsilon,
\end{equation*}
where $\|.\|_1$ is the total variation distance which is defined as the half of the $\ell_1$ distance between two pmfs.
\end{mydef}

\begin{mydef}[Weak Secrecy and Vanishing Error Probability]
 A rate vector \begin{align}\boldsymbol{r}=(r_1,r_2,\cdots,r_t, r_k, r_{k_1}, \cdots, r_{k_t})\end{align} is weakly secure achievable with a vanishing probability of error if for any $\epsilon>0$, there is a code whose rate vectors is in the $\epsilon$ distance of $\boldsymbol{r}$. Furthermore, assuming a uniform and independent distribution over the messages in $\boldsymbol{M}$, the error probability of the code is less than or equal to $\epsilon$ and 
\begin{equation*}
I(\boldsymbol{M};C)\leq\epsilon\cdot H(\boldsymbol{M}).
\end{equation*}
\end{mydef}

It follows from the above definitions that perfect secrecy conditions is stronger than strong secrecy condition, which in turn is stronger than weak secrecy constraint.

\begin{theorem}
\label{Lemma_Equivalent_epsilon_zero}
Assume that $(r_1,r_2,\cdots,r_t, r_k, r_{k_1},r_{k_2},\cdots,r_{k_t})$ is achievable by a sequence of weakly secure codes whose probabilities of error  converge to zero \emph{asymptotically}. We also allow the transmitter to use private randomization in these codes. Then,
\begin{enumerate}
\item [(a)] $(r_1,r_2,\cdots,r_t, r_k, r_{k_1},r_{k_2},\cdots,r_{k_t})$ is achievable by a sequence of perfectly secure and $\epsilon$-error codes.
\item[(b)] There is some $\alpha>0$ such that $\alpha\cdot (r_1,r_2,\cdots,r_t, r_k, r_{k_1},r_{k_2},\cdots,r_{k_t})$ is achievable by a sequence of perfectly secure and zero-error codes, without using private randomization at the transmitter.
\end{enumerate}

\end{theorem}

To prove the Theorems \ref{Main_theorem} and \ref{Lemma_Equivalent_epsilon_zero}, the following lemmas are needed.

\begin{lemma}
\label{Lemma_Eliminate_Private_Keys}
If there exists an $\epsilon$-error perfectly secure code $C$ with the rate vector $$(r_1,r_2,\cdots,r_t, r_k, r_{k_1}, \cdots, r_{k_t}),$$ then $$([r_1-r_{k_1}]_+,[r_2-r_{k_2}]_+,\cdots,[r_t-r_{k_t}]_+, r_k, 0, \cdots, 0)$$ is also $\epsilon$-error perfectly secure achievable.
\end{lemma}

\begin{lemma}
\label{Lemma_Rel_IC_SIC}
Suppose that there is an $\epsilon$-error perfectly secure code $C$ constructed from common key $K$ and messages $M_i$ for $i\in[t]$ where $M_i$ and $K$ are mutually independent uniformly distributed random variables. We assume that no private key $K_i$ is used in the code. Then there is a sequence of conventional codes with zero-error probability whose rate vectors converge to
$$\left(\frac{H(M_1)}{I(\boldsymbol{M};C|K)}, \frac{H(M_2)}{I(\boldsymbol{M};C|K)}, \cdots, \frac{H(M_t)}{I(\boldsymbol{M};C|K)}\right).$$
\end{lemma}

\section{Proofs}
\label{sec:proofs}

\subsection{Proof of Theorem \ref{Main_theorem}}

\textbf{Proof of {(c)$\mapsto$(b)} for both linear and non-linear cases}: Take a conventional index code $C$ and messages $M_i$ achieving rate tuple
$$(\frac{[r_1-r_{k_1}]_+}{r_k}-\epsilon, \frac{[r_2-r_{k_2}]_+}{r_k}-\epsilon, \cdots, \frac{[r_t-r_{k_t}]_+}{r_k}-\epsilon).$$
We construct a new code on the same message sets, and a common keys $K$ on the same alphabet set as $C$, \emph{i.e.,} $|\mathcal{K}|=|\mathcal{C}|$. We use one-time pad and add $C$ with the common key $K$ and broadcast it. The receivers can uncover the original $C$ since they have access to $K$, but it remains hidden from the adversary. Observe that if the original index code is linear, the new index code is also linear.

The rates of the new code is:
\begin{align*}(\frac{[r_1-r_{k_1}]_+}{r_k}-\epsilon, &\frac{[r_2-r_{k_2}]_+}{r_k}-\epsilon, \cdots,\\&\frac{[r_t-r_{k_t}]_+}{r_k}-\epsilon,1,0,0,\cdots, 0)
\\=\alpha\cdot &({[r_1-r_{k_1}]_+}-\epsilon{r_k}, {[r_2-r_{k_2}]_+}-\epsilon{r_k}, \cdots,\\&~{[r_t-r_{k_t}]_+}{r_k}-\epsilon{r_k},r_k,0,0,\cdots, 0),
\end{align*}
where $\alpha=1/r_k$. Letting $\epsilon$ converge to zero, we get the desired result.

\textbf{Proof of {(b)$\mapsto$(a)} for both linear and non-linear cases}: For the non-linear case, it suffices to show that if
$$\alpha\cdot\left(r_1, r_2, \cdots, r_t,r_k,0,0,\cdots, 0\right)\in \mathcal{R}_{\mathsf{Secure}},$$
then for any non-negative $r_{k_1}, \cdots, r_{k_t}$ one can find some $\alpha'>0$ such that
$$\alpha'\cdot(r_1+r_{k_1},r_2+r_{k_2},\cdots,r_t+r_{k_t}, r_k, r_{k_1}, \cdots, r_{k_t})\in \mathcal{R}_{\mathsf{Secure}}.$$
A similar statement is sufficient for the proof of the linear case.
Roughly speaking, the idea is to take a code with messages $M_i$ and a common key $K$. Then we introduce private keys $K_i$ and expand the size of the message $M_i$ by the size of $K_i$. The new $K_i$ bits of $M_i$ are securely transmitted by taking their XOR with the symbols of the private key $K_i$. Again observe that if the original index code was linear, the new index code is also linear.
For a rigorous argument, assume that we start with an index code with public communication $C$. We then have
$\log|\mathcal{M}_i|=\alpha r_{i}\log|\mathcal{C}|$ and
$\log|\mathcal{K}|=\alpha r_{k}\log|\mathcal{C}|$ in the original code. For the new code, we set the size of the messages to be
$\log|\mathcal{M}_i|=\alpha(r_{i}+r_{k_i})\log|\mathcal{C}|$; the size of the common key to be $\log|\mathcal{K}|=\alpha r_{k}\log|\mathcal{C}|$, and the size of private keys to be $\log|\mathcal{K}_i|=\alpha r_{k_i}\log|\mathcal{C}|$. The size of the public communication in the new code that we construct is $\log|\mathcal{C}|+\sum_{i=1}^t\log|\mathcal{K}_i|$, as we are sending $\sum_{i=1}^t\log|\mathcal{K}_i|$ additional XORs. Therefore, the rate tuple of the new code is 
$$\alpha'\cdot(r_1+r_{k_1},r_2+r_{k_2},\cdots,r_t+r_{k_t}, r_k, r_{k_1}, \cdots, r_{k_t})\in \mathcal{R}_{\mathsf{Secure}}$$
where $$\alpha'=\frac{\alpha\log|\mathcal{C}|}{\log|\mathcal{C}|+\sum_{i=1}^t\log|\mathcal{K}_i|}=\frac{\alpha}{1+\sum_{i=1}^tr_{k_i}}.$$

\textbf{Proof of {(b)$\mapsto$(c)} for both linear and non-linear cases}: The linear case is immediate from Theorem \ref{Main_theorem_Linear}. For the non-linear case, we need to show that if
\begin{align*}
\exists \alpha>0:\quad\alpha\cdot(r_1,r_2,\cdots,r_t, r_k, 0, \cdots, 0)\in \mathcal{R}_{\mathsf{Secure}}
\end{align*}
Then \begin{align*}
(\frac{r_1}{r_k}, \frac{r_2}{r_k}, \cdots, \frac{r_t}{r_k})\in \mathcal{R}.\end{align*}
Take a secure index code with messages $M_i$ for $i\in[t]$ and common key $K$ whose rate vector is close to $(r_1,r_2,\cdots,r_t, r_k, 0, \cdots, 0)$. Let $C$ be the public communication of this code. Then $\log|\mathcal{K}|/\log|\mathcal{C}|$ is close to $r_k$ and $\log|\mathcal{M}_i|/\log|\mathcal{C}|$ is close to $r_i$. Hence,
$\log|\mathcal{M}_i|/\log|\mathcal{K}|$ is close to $r_i/r_k$.

Assuming that the messages $M_i$ for $i\in[t]$ and common key $K$ are uniform and mutually independent of each other, we have
\begin{align}
H(\boldsymbol{M})&=H(\boldsymbol{M}|C)+I(\boldsymbol{M};C)\nonumber
\\&=H(\boldsymbol{M}|C)\label{eqM1o}
\\&\leq H(\boldsymbol{M},K|C)\nonumber\\
&=H(\boldsymbol{M}|K,C)+H(K|C)\nonumber\\
&\leq H(\boldsymbol{M}|K,C)+H(K),\nonumber\end{align}
where equality \eqref{eqM1o} comes from perfect secrecy condition. 
Hence,\begin{align}
H(K)&\geq I(\boldsymbol{M};K,C)\nonumber\\
&=I(\boldsymbol{M};C|K)+I(\boldsymbol{M};K)\nonumber\\
&=I(\boldsymbol{M};C|K)\label{eqM2o}.
\end{align}
where equality \eqref{eqM2o} is due to independence of $\boldsymbol{M}$ and $K$.

As our code is zero-error perfectly secure achievable, it is also $\epsilon$-error perfectly secure achievable. Then, by Lemma \ref{Lemma_Rel_IC_SIC}, the rate vector \begin{align}\bigg(\frac{H(M_1)}{I(\boldsymbol{M};C|K)},\frac{H(M_2)}{I(\boldsymbol{M};C|K)},\cdots,\frac{H(M_t)}{I(\boldsymbol{M};C|K)}\bigg)\label{eqn:MCKM}\end{align} belongs to the conventional index coding problem rate region. Therefore, by relation \eqref{eqM2o}, if we replace $I(\boldsymbol{M};C|K)$ by $H(K)$ in equation \eqref{eqn:MCKM}, we get that the rate vector
$$\bigg(\frac{H(M_1)}{H(K)},\frac{H(M_2)}{H(K)},\cdots,\frac{H(M_t)}{H(K)}\bigg)$$
is in the zero-error conventional index coding region. Observe that $\log|\mathcal{M}_i|/\log|\mathcal{K}|$ could be made as close as we desire to $r_i/r_k$. This completes the proof.

We remark that one can have a simpler argument and avoid the use of Lemma \ref{Lemma_Rel_IC_SIC} if the transmitter uses deterministic encoding, \emph{i.e.,} when there is no private randomness and $C$ is a deterministic function of $M$ and $K$. Observe that
\begin{align}
H(K)&\geq I(\boldsymbol{M};C|K)\nonumber\\
&=H(C|K)\label{eqM3o}
\\&\geq \min_{k}H(C|K=k).\nonumber
\end{align}
where inequality \eqref{eqM3o} follows from the fact that $C$ is a function of $(\boldsymbol M, K)$.

If we fix a value of $K=k$, we get a zero-error index code.  Therefore, there exists a zero-error index code whose public communication has length less than or equal to $H(K)=\log|\mathcal{K}|$. The rate vector corresponding to this index code is coordinatewise greater than or equal to
\begin{align*}
(\frac{r_1}{r_k}, \frac{r_2}{r_k}, \cdots, \frac{r_t}{r_k}).\end{align*}
Again as the previous, $\log|\mathcal{M}_i|/\log|\mathcal{K}|$ could be made as close as we desire to $r_i/r_k$, and the proof is concluded.\color{black}

\textbf{Proof of {(a)$\mapsto$(b)}}:

We begin with the linear case, \emph{i.e.,} 
\begin{align*} \exists \alpha>0:\quad \alpha\cdot(r_1,r_2,\cdots,r_t, r_k, r_{k_1}, \cdots, r_{k_t})\in \mathcal{R}_{\mathsf{Secure-Linear}},\end{align*}
implies that \begin{align*}\exists \alpha>0:\quad\alpha\cdot([r_1-r_{k_1}]_+,[r_2-r_{k_2}]_+,\cdots, [r_t-r_{k_t}]_+, r_k, 0, \cdots, 0)\in \mathcal{R}_{\mathsf{Secure-Linear}}.
\end{align*}

Take a sequence of linear secure zero-error index codes with rate vectors approaching 
$$\alpha\cdot(r_1,r_2,\cdots,r_t, r_k, r_{k_1}, \cdots, r_{k_t})$$
for some $\alpha>0$. Let $(l_i, l, l_k, l_{k_i})$ for $i\in[t]$ be a code from this sequence. Then we can apply Theorem \ref{Main_theorem_Linear} to this code to construct a conventional zero-error linear index code with messages of size $[l_i-l_{k_i}]_+$ and $l_k$ symbols of public communication. If we have a secret key of size $l_k$, we can use one-time pad and XOR it with the $l_k$ symbols of public communication. This implies that we can find a secure zero-error index code with messages of size $[l_i-l_{k_i}]_+$, public communication and common key of size $l_k$. This corresponds to the following rate vector
\begin{align*}&\frac{1}{l_k}\cdot([l_1-l_{k_1}]_+,[l_2-l_{k_2}]_+,\cdots, [l_t-l_{k_t}]_+, l_k, 0, \cdots, 0)
=\\&
\frac{l}{l_k}\cdot(\frac{[l_1-l_{k_1}]_+}{l},\frac{[l_2-l_{k_2}]_+}{l},\cdots, \frac{[l_t-l_{k_t}]_+}{l}, \frac{l_k}{l}, 0, \cdots, 0)\end{align*}
which tends to
$$\frac{1}{r_k}\cdot([r_1-r_{k_1}]_+,[r_2-r_{k_2}]_+,\cdots, [r_t-r_{k_t}]_+, r_k, 0, \cdots, 0).$$

This completes the proof for the linear case. Next, we consider the general non-linear case. We need to show that
\begin{align*}\exists \alpha>0:\quad \alpha\cdot(r_1,r_2,\cdots,r_t, r_k, r_{k_1}, \cdots, r_{k_t})\in \mathcal{R}_{\mathsf{Secure}},\end{align*}
implies that
\begin{align*}\exists \alpha>0:\quad\alpha\cdot([r_1-&r_{k_1}]_+,[r_2-r_{k_2}]_+,\cdots,\\& [r_t-r_{k_t}]_+, r_k, 0, \cdots, 0)\in \mathcal{R}_{\mathsf{Secure}}.
\end{align*}

As the rate vector $(r_1,r_2,\cdots,r_t, r_k, r_{k_1}, \cdots, r_{k_t})$ is zero-error perfectly secure achievable, it is also $\epsilon$-error perfectly secure achievable. Then, using Lemma \ref{Lemma_Eliminate_Private_Keys}, by eliminating private keys, the rate vector $([r_1-r_{k_1}]_+,[r_2-r_{k_2}]_+,\cdots,[r_t-r_{k_t}]_+, r_k, 0, \cdots, 0)$ is $\epsilon$-error perfectly secure achievable, too. We have constructed a code with \emph{asymptotically} zero probability of error, not exactly zero probability of error as required in our model. To complete the proof, one is needed to prove that if $(r_1,r_2,\cdots,r_t,r_k,0,\cdots,0)$ is $\epsilon$-error perfectly secure achievable, there exists $\alpha$ so that $\alpha\cdot(r_1,r_2,\cdots,r_t,r_k,0,\cdots,0)$ is perfectly secure zero-error achievable. But this follows from part (b) of Theorem \ref{Lemma_Equivalent_epsilon_zero}.


\subsection{Proof of Theorem \ref{Main_theorem_Linear}}
Assume that there exists a zero-error secure linear index code $C$ with key lengths $l_k,l_{k_i} (i\in[t])$ and private randomness of length $l_w$. We assume that $l$ equations are created by the transmitter from the message symbols and the private and public keys. Without loss of generality, we can assume that there is no zero-error secure index code $C'$ with 
$$(l'_1, \cdots, l'_t)=(l_1, \cdots, l_t)$$
but $l'\leq l$, $l'_k\leq l_k$, $l'_{k_i}\leq l_{k_i}$, $l'_w\leq l_w$ and
$$(l', l'_k, l'_{k_1}, \cdots, l'_{k_t}, l'_w)\neq (l, l_k, l_{k_1}, \cdots, l_{k_t}, l_w).$$
We refer to this as the minimality assumption. It implies that the code matrix $\Pi$ given in equation \eqref{eqn5} has no all-zero column and the matrix $\Pi$ is full row rank. Otherwise, there exists a key bit or a message bit which has not been used in producing the index code, or the length of the index code could be reduced. 

Our goal is to show that the minimality assumption implies that $l_w=0$ and furthermore one can use elementary row and other valid operations to convert the code matrix $\Pi$ to the following form, while preserving decodability and security of the code.\vspace{-0.5cm}
\begin{equation}
\bordermatrix{&&&&\cr
                &\Lambda^{(0)} &  \boldsymbol{0}&\boldsymbol{0}&\ldots & \boldsymbol{0}&\VR &&\cr
               & \boldsymbol{0}  &\Lambda^{(1)} &\boldsymbol{0} &\ldots & \boldsymbol{0}&\VR &&\cr
                & \boldsymbol{0} & \boldsymbol{0} & \Lambda^{(2)} &\cdots& \boldsymbol{0}&\VR&{\Scale[2]{\Gamma}}&\cr
                & \vdots & \vdots &\vdots&\ddots & \vdots&\VR&&\cr
                & \boldsymbol{0} & \boldsymbol{0} &\boldsymbol{0}&\ldots & \Lambda^{(t)}&\VR &&},\label{Secure_IC_One_Shot_Linear}
\end{equation}
where  $\Lambda^{(0)}=I_{l_k\times l_k}$, $\Lambda^{(i)}=I_{l_{k_i}\times l_{k_i}}$ are identity matrices, and $\Gamma$ is a $l\times(\sum_{i=1}^tl_i)$ submatrix, which gets multiplied by the message vector. This statement implies, in particular, that the number of rows of matrix $\Pi$ should be equal to $l=l_k+\sum_{i=1}^tl_{k_i}$.

 With  elementary row operations, we bring the matrix $\Pi$ in its row echelon form, calling it $\widetilde{\Pi}$. Since the operations are invertible, the decodability and reliability constraints are preserved. Since $\Pi$ was full row rank, $\widetilde{\Pi}$ does not have an all-zero row. By the minimality assumption, we also do not have an all-zero column in $\widetilde{\Pi}$. 

Each row of $\widetilde{\Pi}$ has the form $[0~0~\cdots~0~1~*~*~\cdots *]$. The symbol $1$ appearing in this row cannot correspond to a message symbol since the equation for this row will then correspond to a linear combination of only message symbols, which is a contradiction with the security assumption (observe that in equation \eqref{eqn5}, message symbols come at the end of the vector). Therefore, the symbol $1$ should correspond to either $K(i)$ or $K_j(i)$ or $W(i)$ for some $i$. We call a coordinate of $K$, $K_j$ or $W$ to be \emph{marked} if it corresponds to a symbol $1$ appearing as the first non-zero element of a row of $\widetilde{\Pi}$. Observe that each coordinate of $K$, $K_j$ or $W$ that is marked occurs only in one row of $\widetilde{\Pi}$ because of its row echelon form. 

We claim that all coordinates of $K$ and $K_j$ and $W$ are marked. Otherwise, if for instance $K(i)$ is not marked for some $i$, we can fix it to be zero (effectively reducing the length of $K$ by one). Decoding is still possible, since we had that given any arbitrary choice for $K(i)$, decoding is possible; hence decoding is possible when $K(i)$ is fixed to be zero for some $i$. The new code is also secure since every equation contains a marked element of one of the vectors $K$, $K_j$ and $W$, and that element occurs in only and only that equation. Presence of these marked elements make the equations secure from the perspective of the adversary, as in one-time pad (mask the equations). 
Thus, the minimality assumption implies that all coordinates of $K$ and $K_j$ and $W$ are marked, and $\widetilde{\Pi}$ has the following form
\begin{equation}
\bordermatrix{&&&&\cr
                &\Lambda^{(0)} &  \boldsymbol{0}&\boldsymbol{0}&\ldots & \boldsymbol{0}&\boldsymbol{0}&\VR &&\cr
               & \boldsymbol{0}  &\Lambda^{(1)} &\boldsymbol{0} &\ldots & \boldsymbol{0}&\boldsymbol{0}&\VR &&\cr
                & \boldsymbol{0} & \boldsymbol{0} & \Lambda^{(2)} &\cdots& \boldsymbol{0}&\boldsymbol{0}&\VR&{\Scale[2]{\Gamma}}&\cr
                & \vdots & \vdots &\vdots&\ddots & \vdots&\vdots&\VR&&\cr
                & \boldsymbol{0} & \boldsymbol{0} &\boldsymbol{0}&\ldots & \Lambda^{(t)}&\boldsymbol{0}&\VR &&\cr,
                & \boldsymbol{0} & \boldsymbol{0} &\boldsymbol{0}&\ldots & 0&\Lambda^{(t+1)}&\VR &&},\label{Secure_IC_One_Shot_Linear2}
\end{equation}
where  $\Lambda^{(0)}=I_{l_k\times l_k}$, $\Lambda^{(i)}=I_{l_{k_i}\times l_{k_i}}$ and $\Lambda^{(t+1)}=I_{l_{w}\times l_{w}}$ are identity matrices. Now, observe that the equations that are marked by coordinates of $W$ are masked from all the receivers, as well as the adversary (each of these equations including the XOR with one and only one of the elements of $W$). Therefore, they are not useful in decoding of the messages by the receivers and can be removed. This implies that $l_w=0$, and we get that $\widetilde{\Pi}$ is in the desired form given in equation \eqref{Secure_IC_One_Shot_Linear}.

We have shown that corresponding to any arbitrary linear zero-error perfectly secure code, there is another linear zero-error perfectly secure index code for the same message sets that uses secret keys of lengths $(l_k, l_{k_1}, \cdots, l_{k_t})$ with the following property: each of the $l$ symbols of the public message are of the form
\begin{align}C_i=K(p)+\sum_{j=1}^t\sum_{p=1}^{l_j}\gamma^{i}_{jp}M_{j}(p)\label{eqn:pub1}\end{align}
for some $p\in[l_k]$, or
\begin{align}C_i=K_{i}(p)+\sum_{j=1}^t\sum_{p=1}^{l_j}\gamma^{i}_{jp}M_{j}(p)\label{eqn:priv1}\end{align}
for some $i\in[t]$ and $p\in[l_{k_i}]$. In other words, the expression of each of the code symbols $C_i$   contains only one symbol from one of the secret keys.

Consider the first receiver. It has access to $l$ linear equations of the form given in \eqref{eqn:pub1} (as it has $K$), and $l_1$ linear equations of the form given in \eqref{eqn:priv1} (as it has $K_1$). Therefore, we call the $l$ equations as public to all receiver, and the $l_1$ equations as private to the receiver one. We now use Lemma \ref{lemma2} with $X=M_1$ and $Y=(M_2, M_3, \cdots, M_t)$, $AX+BY$ being equations of the form  given in \eqref{eqn:pub1}, and $CX+DY$ being the equations of the form given in \eqref{eqn:priv1}. This lemma then implies that there is a subset of the entries of $M_1$ of size at most $l_1$ such that from the values of these entries and the $l$ public equations, receiver one can recover $M_1$. Let us fix $M_1$ on these $l_1$ locations and reveal its value to all the receivers. The number of free entries of $M_1$, \emph{i.e.,} the new length of the message of $M_1$, would then be greater than or equal to $l-l_1$. This message can be decoded by the first receiver using the $l$ public linear equations of the form given in \eqref{eqn:pub1}. The fact that we have fixed some of entries of $M_1$ and given it to other receivers can only help them recover their messages (because if they did not know $M_1$, we are giving them some partial information about $M_1$). A similar procedure can be done for other receivers. This would imply that with $l$ linear equations, it is possible for the receiver $i$ to recover $l-l_i$ symbols using $l$ public symbols of message. This is the claim we wanted to prove. The proof is complete.

\begin{lemma}\label{lemma2}
Let $X_{1\times n}$ and $Y_{1\times m}$ be two arbitrary column vectors in a field $\mathbb{F}$. Assume that matrices $A_{l\times n}$, $B_{l\times m}$, $C_{l_1\times n}$ and $D_{l_1\times m}$ are such that the vector $X$ can be recovered from the values of $AX+BY$ and $CX+DY$. Then, there is a subset of indices $\mathcal S\subset [n]$ with $|\mathcal S|\leq l_1$, such that it is possible to find $X$ from $AX+BY$ and $X(i), i\in \mathcal S$. Here $X(i)$ is used to denote the $i$-th entry of vector $X$.
\end{lemma}

\begin{proof}
Consider the first row of $CX+DY$, which is a linear equation in terms of the entries of $X$ and $Y$, say $\sum\alpha_iX(i)+\sum\beta_jY(j)$. If we can find $X$ without having access to this row, we discard it and proceed to the second row. Otherwise, there is an entry of $X$, say $i_1$ that cannot be decoded without the linear equation $\sum\alpha_iX(i)+\sum\beta_jY(j)$. In other words, $X(i_1)$ is a linear combination of the linear equations that we have, with the equation $\sum\alpha_iX(i)+\sum\beta_jY(j)$ being given a non-zero weight. Then if we put $i_1$ in the set $\mathcal{S}$ of the entries that we know, we can conversely use it to recover the linear equation $\sum\alpha_iX(i)+\sum\beta_jY(j)$. Therefore, having $X(i_1)$ is equivalent to having $\sum\alpha_iX(i)+\sum\beta_jY(j)$. Continuing with this procedure, we can construct the set $\mathcal S$ and its size will be less than or equal to the number of rows of $CX+DY$, which is $l_1$. 
\end{proof}


\subsection{Proof of Theorem \ref{Lemma_Equivalent_epsilon_zero}}
\subsubsection{Proof of part (a)}
The proof of part (a) consists of two steps.  We first show the rate region equivalency of $\epsilon$-error strongly secure code to the $\epsilon$-error perfectly secure code. Then, we say that if a rate region is $\epsilon$-error weakly secure achievable, it is also $\epsilon$-error strongly secure achievable.

\emph{From Strong to Perfect Secrecy for Free:} We are supposing a strong secrecy condition, \emph{i.e.,} the independence between $\boldsymbol{M}$ and $C$ no longer exists, and instead, the following inequality holds:
\begin{equation*}
\|p(\boldsymbol{m},c)-p(\boldsymbol{m})p(c)\|_1\leq\epsilon.
\end{equation*}
We would like to make $I(\boldsymbol{M};C)=0$, without using 
 additional key bits. Using the coupling method, one can find $\boldsymbol{M}', C'$ having the marginal pmf $p(\boldsymbol{m})p(c)$ and jointly distributed $\boldsymbol{M}, C$ with such that 
$$p\big((\boldsymbol{M}, C)\neq (\boldsymbol{M}', C')\big)\leq \|p(\boldsymbol{m},c)-p(\boldsymbol{m})p(c)\|_1\leq \epsilon.$$
Let $p_{\boldsymbol{M}, C, \boldsymbol{M}',C'}$ denote the induced joint distribution by the coupling method. Observe that $\boldsymbol{M}'$ has the uniform marginal distribution $p(m)$ and is independent of $C'$. The encoder proceeds as follows: the encoder assumes $\boldsymbol{M}'$ to be the intended messages to the receivers, produces $\boldsymbol{M}, C, C'$ via the conditional distribution $p_{\boldsymbol{M}, C, C'|\boldsymbol{M}'}$ and broadcasts $C'$. We have perfect secrecy as $C'$ is independent of $\boldsymbol{M}'$. Since with probability $1- \epsilon$, random variables $\boldsymbol{M}', C'$ are equal to $\boldsymbol{M}, C$, the total error probability will be increased by at most $\epsilon$ that can be made arbitrarily small. This completes the proof.


\emph{From Weak to Strong Secrecy for Free:} Suppose we have a code $C$ satisfying the weak secrecy condition, \emph{i.e.,} $I(\boldsymbol{M};C)\leq\epsilon\cdot H(\boldsymbol{M})$, and error probability $\epsilon$. From Fano's inequality, we have $H(\boldsymbol{M}|\hat{\boldsymbol{M}})\leq \delta$, where $\hat{\boldsymbol{M}}$ is the vector of reconstructions by the decoders and $\delta=h(\epsilon)+\epsilon\log|\boldsymbol{\mathcal{M}}|$. 

Consider $n$ i.i.d.~repetitions of the code. Assuming that $R_i=\log|\mathcal M_i|$, we get $\lvert\mathcal{M}_i^n\rvert=2^{nR_i}$. Let $$\tilde{R}_i=R_i-2\epsilon\cdot H(\boldsymbol{M})-2\delta\cdot t, \qquad \bar{R}_i=2\delta$$ where $t$ is the number of nodes. We randomly and independently bin $\mathcal{M}_i^n$ into $2^{n\tilde{R}_i}$ and $2^{n\bar{R}_i}$ bins for $i\in [t]$, and denote the bin indices by $\widetilde{M}_i$ and $\bar{M}_i$.  Theorem 1 of \cite{yaar14} provides sufficient condition for the following to hold: for any given $\eta>0$, there exists an integer $n$ and such that 
\begin{equation}
\label{eq:total_var_strong_sec}
\mathbb{E}\lVert P_{\widetilde{\boldsymbol{M}}\bar{\boldsymbol{M}}C^n}-p_{\widetilde{\boldsymbol{M}}}^Up_{\bar{\boldsymbol{M}}}^U p_{C^n}\rVert\leq \eta
\end{equation}
where the expected value is over all random binning indices and $p^U$ is the uniform distribution. The sufficient condition for the above to hold is that 
 for each $\mathcal{S}\subseteq [t]$, the binning rate vector $(\tilde{R}_1, \bar{R}_1,\tilde{R}_2, \bar{R}_2,\cdots,\tilde{R}_t, \bar{R}_t)$ satisfies the following inequality,
\begin{equation}
\label{OSRB_Rate_Condition}
\sum_{i\in\mathcal{S}}{\tilde{R}_i+\bar{R}_i}< H(M_{\mathcal{S}}|C)=H(M_{\mathcal{S}})-I(M_{\mathcal{S}};C)=\sum_{i\in\mathcal{S}}{R_i}-I(M_{\mathcal{S}};C).
\end{equation}
Furthermore, by the Slepian-Wolf theorem, we can recover $M_i^n$ from $(\hat{M}_i^n, \bar{M_i})$ with probability $1-\epsilon$ (for $n$ sufficiently large) for each $i\in [t]$ if
\begin{equation}
\label{SW_Rate_Condition}
\bar{R}_i> H(M_i|\hat{M}_i), \qquad \forall i\in[t].
\end{equation}

If equations \eqref{OSRB_Rate_Condition} and \eqref{SW_Rate_Condition} hold, one can find a deterministic binning such that 
\begin{equation}
\label{eq:total_var_strong_sec22}
\lVert p_{\widetilde{\boldsymbol{M}}\bar{\boldsymbol{M}}C^n}-p_{\widetilde{\boldsymbol{M}}}^Up_{\bar{\boldsymbol{M}}}^U p_{C^n}\rVert\leq \eta
\end{equation}
holds and furthermore, with probability $1-\epsilon$, $M_i^n$ can be recovered from $(\hat{M}_i^n, \bar{M_i})$.

We claim that  equations \eqref{OSRB_Rate_Condition} and \eqref{SW_Rate_Condition} hold for our choice of $\tilde{R}_i=R_i-2\epsilon\cdot H(\boldsymbol{M})-2\delta\cdot t$ and $\bar{R}_i=2\delta$. Observe that the right hand of the inequality \eqref{SW_Rate_Condition} is less than or equal to $h(\epsilon)+\epsilon\log|\mathcal{M}_i|$ which is itself less than or equal to $\delta$. To verify equation \eqref{OSRB_Rate_Condition}, we utilize the fact that the right hand of the inequality \eqref{OSRB_Rate_Condition} is greater than $\sum_{i\in\mathcal{S}}{R_i}-\epsilon\cdot H(\boldsymbol{M})$ by the assumption of weak secrecy.

Equation \eqref{eq:total_var_strong_sec22} implies that we have strong security if we take $(C^n, \bar{\boldsymbol{M}})$ as the public message for the new code and take $\tilde{M}_i$ as the messages, we wish to transmit. The fact that $M_i^n$ can be recovered from $(\hat{M}_i^n, \bar{M_i})$ implies that the $i$-th node is able to use $C^n$ to first find $\hat{M}_i^n$ and then $\bar{M_i}$ to recover $M_i^n$ with probability $1-\epsilon$. Then, from $\hat{M}_i^n$, the node can recover its message  $\tilde{M}_i$ as its bin index. The overall error probability will be at most $t\epsilon$ by the union bound.

We should only note that here the messages $\tilde{M}_i$ are almost uniform and mutually independent, as from \eqref{eq:total_var_strong_sec22}, we have
$$\lVert p_{\widetilde{\boldsymbol{M}}}-p_{\widetilde{\boldsymbol{M}}}^U\rVert\leq \eta.$$
But using the coupling method, as in the previous part, we can couple $\widetilde{\boldsymbol{M}}$ with a mutually independent and uniform messages $\widetilde{\boldsymbol{M}}'$ such that $\widetilde{\boldsymbol{M}}=\widetilde{\boldsymbol{M}}'$ with high probability. Therefore, we can impose the uniformity and independence constraint by slightly increasing the error probability of the code, and while preserving the strong security constraint.

The rate of the original code was
$$r_i=\frac{\log|\mathcal M_i|}{\log|\mathcal{C}|}=\frac{R_i}{\log|\mathcal{C}|}.$$
Rate of the new code is
\begin{align*}\tilde r_i&=\frac{\tilde{R}_i}{\log|\mathcal{C}|+\bar{R}_i}
\\&=\frac{R_i-2\epsilon\cdot H(\boldsymbol{M})-2\delta\cdot t}{\log|\mathcal{C}|+2\delta}
\\&=\frac{R_i-2\epsilon\cdot \sum_{i=1}^t R_i-2(h(\epsilon)+\epsilon\sum_{i=1}^t R_i)\cdot t}{\log|\mathcal{C}|+2(h(\epsilon)+\epsilon\sum_{i=1}^t R_i)}
\\&=\frac{r_i-2\epsilon\cdot \sum_{i=1}^t r_i-2(v+\epsilon\sum_{i=1}^t r_i)\cdot t}{1+2(v+\epsilon\sum_{i=1}^t r_i)}
\end{align*}
where $v=h(\epsilon)/\log|\mathcal{C}|\leq h(\epsilon)$. Letting $\epsilon$ converge to zero, we get that $\tilde{r}_i\rightarrow r_i,~i\in[t]$.


\subsubsection{Proof of part (b)}
{\color{black} We would like to show that if a rate vector $(r_1,r_2,\cdots,r_t,r_k,r_{k_1},r_{k_2},\cdots,r_{k_t})$ is $\epsilon$-error perfectly secure achievable, then there exist some positive multiplicative constant $\alpha$ so that $$\alpha\cdot (r_1,r_2,\cdots,r_t,r_k,r_{k_1},r_{k_2},\cdots,r_{k_t})$$ could be achieved by zero-error perfectly secure codes. By Lemma \ref{Lemma_Eliminate_Private_Keys}, the $\epsilon$-error perfectly secure achievability of $(r_1,r_2,\cdots,r_t,r_k,r_{k_1},r_{k_2},\cdots,r_{k_t})$ leads to the $\epsilon$-error perfectly secure achievability of $([r_1-r_{k_1}]_+,[r_2-r_{k_2}]_+,\cdots,[r_t-r_{k_t}]_+, r_k, 0, \cdots, 0)$. In the following, we show that there exist some $\alpha>0$ so that $\alpha\cdot([r_1-r_{k_1}]_+,[r_2-r_{k_2}]_+,\cdots,[r_t-r_{k_t}]_+, r_k, 0, \cdots, 0)$ is zero-error perfectly secure achievable. 
This claim would establish the desired result by using part $(b)\mapsto (a)$ of the Theorem \ref{Main_theorem} and adding back the private keys. 

Therefore, it remains to show that there exist some $\alpha>0$ so that $\alpha\cdot([r_1-r_{k_1}]_+,[r_2-r_{k_2}]_+,\cdots,[r_t-r_{k_t}]_+, r_k, 0, \cdots, 0)$ is zero-error perfectly secure achievable. To proceed, it suffices to show that if $(r_1,r_2,\cdots,r_t,r_k,\\0,\cdots,0)$ is achievable by a sequence of codes with vanishing probability of error and perfect secrecy conditions, there exist some $\alpha>0$ so that $\alpha\cdot(r_1,r_2,\cdots,r_t,r_k,0,\cdots,0)$ is zero-error perfectly secure achievable. 

To do this, take an $\epsilon$-error code with corresponding variables $K, C$, and $M_i$ for $i\in[t]$ where $M_i$ and $K$ are uniform and mutually independent random variables. Also let $\widehat{M}_i$ to be the reconstruction by receiver $i$. Since private randomization at the transmitter is allowed, $C$ is not necessarily a deterministic function of $(K, \boldsymbol{M})$.}

As before, we have
\begin{align}
H(\boldsymbol{M})&=H(\boldsymbol{M}|C)+I(\boldsymbol{M};C)\nonumber
\\&=H(\boldsymbol{M}|C)\label{eqM1}
\\&\leq H(\boldsymbol{M},K|C)\nonumber\\
&=H(\boldsymbol{M}|K,C)+H(K|C)\nonumber\\
&\leq H(\boldsymbol{M}|K,C)+H(K),\end{align}
where equality \eqref{eqM1} comes from perfect secrecy condition. Hence,\begin{align}
H(K)&\geq I(\boldsymbol{M};K,C)\nonumber\\
&=I(\boldsymbol{M};C|K)+I(\boldsymbol{M};K)\nonumber\\
&=I(\boldsymbol{M};C|K)\label{eqM2}
\end{align}
where equality \eqref{eqM2} is due to independence of $\boldsymbol{M}$ and $K$. Hence $H(K)\geq I(\boldsymbol{M};C|K)$. Thus, the rate vector of the code is 
\begin{align*}\left(\frac{H(M_1)}{\log|\mathcal{C}|}, \right.&\left.\frac{H(M_2)}{\log|\mathcal{C}|}, \cdots, \frac{H(M_t)}{\log|\mathcal{C}|}, \frac{H(K)}{\log|\mathcal{C}|}, 0, 0, \cdots, 0\right)=\\
\frac{I(\boldsymbol{M};C|K)}{\log|\mathcal{C}|}&\left(\frac{H(M_1)}{I(\boldsymbol{M};C|K)}, \frac{H(M_2)}{I(\boldsymbol{M};C|K)}, \cdots,\right.
\\&\qquad \left.\frac{H(M_t)}{I(\boldsymbol{M};C|K)}, \frac{H(K)}{I(\boldsymbol{M};C|K)}, 0, 0, \cdots, 0 \right)
\end{align*}
The term ${I(\boldsymbol{M};C|K)}/{\log|\mathcal{C}|}$ is a multiplicative factor. 
Since ${H(K)}/{I(\boldsymbol{M};C|K)}\geq 1$ from equation \eqref{eqM2}, to show that we can reach the rate vector
\begin{align*}&\left(\frac{H(M_1)}{I(\boldsymbol{M};C|K)}, \frac{H(M_2)}{I(\boldsymbol{M};C|K)}, \cdots,\right.
\\&\qquad \left.\frac{H(M_t)}{I(\boldsymbol{M};C|K)}, \frac{H(K)}{I(\boldsymbol{M};C|K)}, 0, 0, \cdots, 0 \right)\end{align*}
with  perfectly secure zero-error codes, it suffices to show that there is a sequence of perfectly secure zero-error codes whose rate vectors converge to
\begin{align*}&\left(\frac{H(M_1)}{I(\boldsymbol{M};C|K)}, \frac{H(M_2)}{I(\boldsymbol{M};C|K)}, \cdots,\frac{H(M_t)}{I(\boldsymbol{M};C|K)}, 1, 0, 0 , \cdots, 0\right).\end{align*}
But the rate of $r_k=1$ means that the size of common key and public communication are the same. Therefore one can always use one-time pad to ensure perfect security. It only remains to show that there is a sequence of conventional zero-error codes whose rate vectors converge to
$$\left(\frac{H(M_1)}{I(\boldsymbol{M};C|K)}, \frac{H(M_2)}{I(\boldsymbol{M};C|K)}, \cdots, \frac{H(M_t)}{I(\boldsymbol{M};C|K)}\right).$$
But this follows from Lemma \ref{Lemma_Rel_IC_SIC}.


\subsection{Proof of Lemma \ref{Lemma_Eliminate_Private_Keys}}
We need to show that if $(r_1,r_2,\cdots,r_t, r_k, r_{k_1}, \cdots, r_{k_t})$ is $\epsilon$-error perfectly secure achievable, by eliminating private keys, $([r_1-r_{k_1}]_+,[r_2-r_{k_2}]_+,\cdots,[r_t-r_{k_t}]_+, r_k, 0, \cdots, 0)$ is $\epsilon$-error perfectly secure achievable.

Take an arbitrary index code $C, K, M_i$ and $K_i$ for $i\in[t]$. We create a new secure index code that does not have private keys and is able to securely and reliably achieve message rates $(\log|\mathcal{M}_i|-\log|\mathcal{K}_i|)/\log|\mathcal{C}|$ for $i\in[t]$ and the same common key rate  $\log|\mathcal{K}|/\log|\mathcal{C}|$. This would conclude the proof.

In the original code, we assume that $M_i$'s, $K$ and $K_i$'s are mutually independent. Let us now consider a different scenario where the receivers do not have access to $K_i$'s. In other words, $K_i$ for $i\in[t]$ is simply treated as a private randomness of the transmitter. Thus, only the common key is shared with the legitimate receivers and the private keys, $K_i$, are not available at the receivers. Fig. \ref{fig:Secure_IC_lemma} illustrates the secure index coding scheme by ignoring the private keys in the receivers. In the figure we use $\bar{Y_i}$ to denote the total information available at the receiver $i$ when $K_i$'s are not available. Here, the adversary cannot learn anything about the messages. However, the problem is that the legitimate receivers cannot decode their intended messages.

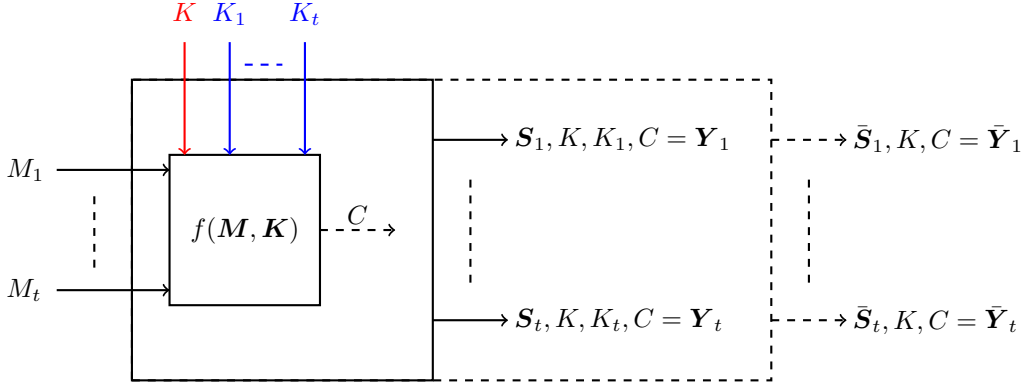
\begin{figure}[t]
\centering
\begin{tikzpicture}[scale=1]

\draw[thick] (0,0) rectangle (4,4);
\draw[thick,dashed] (0,0) rectangle (8.5,4);
\draw[thick] (.5,1) rectangle (2.5,3) node [black,inner sep=10pt, pos=.5, anchor=center] {$f(\boldsymbol{M},\boldsymbol{K})$};

\draw[thick,->] (-1,2.8) -- (.5,2.8) node [black,inner sep=2pt, pos=-.5, anchor=west] {${M_1}$};
\draw[thick,->] (-1,1.2) -- (.5,1.2) node [black,inner sep=2pt, pos=-.5, anchor=west] {${M_t}$};
\draw[thick,dashed] (-.5,1.5) -- (-.5,2.5);

\draw[thick,->,color=red] (.7,4.5) -- (.7,3) node [black,inner sep=2pt, pos=-.4, anchor=north] {\color{red}$K$};
\draw[thick,->,color=blue] (1.3,4.5) -- (1.3,3) node [black,inner sep=2pt, pos=-.4, anchor=north] {\color{blue}$K_1$};
\draw[thick,->,color=blue] (2.3,4.5) -- (2.3,3) node [black,inner sep=2pt, pos=-.4, anchor=north] {\color{blue}$K_t$};
\draw[thick,dashed,color=blue] (1.5,4.2) -- (2,4.2);

\draw[thick,->] (4,3.2) -- (5,3.2)node [black,inner sep=2pt, pos=1, anchor=west] {${\boldsymbol{S}_1}, K,K_1, C={\boldsymbol{Y}_1}$};
\draw[thick,->] (4,.8) -- (5,.8)node [black,inner sep=2pt, pos=1, anchor=west] {${\boldsymbol{S}_t}, K,K_t, C={\boldsymbol{Y}_t}$};
\draw[thick,dashed] (4.5,1.3) -- (4.5,2.7);

\draw[thick,dashed,->] (8.5,3.2) -- (9.5,3.2)node [black,inner sep=2pt, pos=1, anchor=west] {${\bar{\boldsymbol{S}}_1}, K, C={\bar{\boldsymbol{Y}}_1}$};
\draw[thick,dashed,->] (8.5,.8) -- (9.5,.8)node [black,inner sep=2pt, pos=1, anchor=west] {${\bar{\boldsymbol{S}}_t}, K, C={\bar{\boldsymbol{Y}}_t}$};
\draw[thick,dashed] (9,1.3) -- (9,2.7);

\draw[thick,dashed,->] (2.5,2) -- (3.5,2)node [black,inner sep=2pt, pos=.5, anchor=south] {$C$};
\end{tikzpicture}
\caption{The schematic of secure index coding scenario in which the private keys $K_i$'s are not available at the receivers.} 
\label{fig:Secure_IC_lemma}
\end{figure}

We construct a $t$-input, $t$-output interference channel as follows: the input of the $i$-th transmitter is $M_i$, and the output of the $i$-th receiver is $\bar{Y_i}$. Using the result of \cite[p. 133]{Elgamal2011} by treating interference as noise, rates $(R_1, \cdots, R_t)$ is asymptotically achievable with repeated use of this interference channel, if $R_i\leq I(M_i;\bar{Y_i})$. Observe that
\begin{align*}
I(M_i;\bar{Y_i})&=I(M_i;\bar{Y_i},K_i)-I(M_i;K_i|\bar{Y_i})\\
&=I(M_i;Y_i)-I(M_i;K_i|\bar{Y_i})\\
&\overset{(a)}{\geq}H(M_i)-h(\epsilon)-\epsilon\cdot\log|\mathcal{M}_i|-I(M_i;K_i|\bar{Y_i})\\
&\geq H(M_i)-H(K_i)-h(\epsilon)-\epsilon\cdot\log|\mathcal{M}_i|,
\\&=\log|\mathcal{M}_i|-\log|\mathcal{K}_i|-h(\epsilon)-\epsilon\cdot\log|\mathcal{M}_i|,
\end{align*}
where $(a)$ follows from Fano's inequality and the fact that $Y_i$ gives an $\epsilon$-error approximate of $M_i$. In other words, as the receiver $i$ using $Y_i$ can recover $M_i$ with the $\epsilon$ probability of error, $I(M_i;Y_i)$ is approximately equal to $H(M_i)$. Moreover, $h(\epsilon)$ is the binary entropy.

Therefore, messages of rates $H(M_i)-H(K_i)$ can be sent with $N$ uses of the original code. The input distribution on $M_i^N$ will be uniform over the codewords, which is no longer uniform. However, the adversary would not learn anything about the messages since  perfect security constraint holds as long as the common key is uniform and mutually independent of the messages; the marginal distribution of the messages is not important (see equation \eqref{eqn:addedM} and the justification given for it).  Hence, using the constructed code $C$, we could achieve the rate vector $([r_1-r_{k_1}]_+,[r_2-r_{k_2}]_+,\cdots,[r_t-r_{k_t}]_+, r_k, 0, \cdots, 0)$ with \emph{asymptotically} zero probability of error and perfect secrecy.


\subsection{Proof of Lemma \ref{Lemma_Rel_IC_SIC}}

Consider a secure $\epsilon$-error code with corresponding variables $K, C$, and $M_i$ for $i\in[t]$ where $M_i$ and $K$ are uniform and mutually independent random variables. 
It has been shown in \cite{laneff11} that in the conventional index coding, zero and asymptotic error capacities are exactly the same. Therefore, we need to show that there is a sequence of conventional vanishing error codes whose rate vectors converge to
$$\left(\frac{H(M_1)}{I(\boldsymbol{M};C|K)}, \frac{H(M_2)}{I(\boldsymbol{M};C|K)}, \cdots, \frac{H(M_t)}{I(\boldsymbol{M};C|K)}\right).$$

From the perspective of the legitimate parties $K$ is a common randomness, independent of the messages. We assume that the receiver $i$ uses decoding function, as in equation \eqref{eqn:defg-i1},
$$g_i: \mathcal{C}\times \mathcal{S}_i\times\mathcal{K}\rightarrow \mathcal{M}_i,$$ to produce $\hat{M}_i$.

The above code induces a joint distribution $p(\boldsymbol{M},C,K,\widehat{\boldsymbol{M}})$. Let us take $n$ i.i.d.~ repetitions of $(\boldsymbol{M}, K)$. We would like to use the covering lemma \cite[Lemma 3.3]{Elgamal2011}. If $R=I(\boldsymbol{M};C|K)+\epsilon'$, there is a codebook $\hat C_{k^n}^n(1), \hat C_{k^n}^n(2), \cdots, \hat C_{k^n}^n(2^{nR})$ of sequences in $\mathcal{C}^n$ for each $k^n$,  such that with high probability, given $k^n, \boldsymbol{m}^n$, one can find an index $j$ where $(\hat C_{k^n}^n(j), k^n, \boldsymbol{m}^n)$ are jointly typical according to  $p(C,K,\boldsymbol{M})$. 

Now, let us construct a conventional index code (no secrecy) with messages $M_i^n$ for $i\in[t]$ and a shared common randomness $K^n$ among all the parties. Having observed $(k^n, \boldsymbol{m}^n)$, the transmitter finds an index $j$ where $(\hat C_{k^n}^n(j), k^n, \boldsymbol{m}^n)$ are jointly typical. Index $j$ is sent over the public channel. Sending this index requires only $I(\boldsymbol{M};C|K)+\epsilon'$ bits on average. Let us denote $\hat C_{k^n}^n(j)$ by $c^n$. Now, receiver $i$ gets a sequence $c^n$, the common randomness $K^n$ and its side information about other user's messages. The decoder applies $n$ copies of the same decoding function $g_k(\cdot)$ to the sequences $c^n, k^n$ and its side information about the messages (as if we were running $n$ identical copies of the original code and $c^n$ was $n$ copies of the message from the $n$ instances of the code). This results in reconstructions $\widehat{\boldsymbol{m}}^n$ that is jointly typical with $(c^n, k^n, \boldsymbol{m}^n)$ with high probability according to $p(\boldsymbol{M},C,K,\widehat{\boldsymbol{M}})$. This implies that in particular, $(\widehat{\boldsymbol{m}}^n, \boldsymbol{m}^n)$ will be jointly typical according to  $p(\boldsymbol{M},\widehat{\boldsymbol{M}})$ with high probability. But since in the pmf induced by the code, error probability $P(\boldsymbol{M}\neq \widehat{\boldsymbol{M}})\leq \epsilon$, $(\widehat{\boldsymbol{m}}^n, \boldsymbol{m}^n)$ are jointly typical only if  $\widehat{\boldsymbol{m}}(j)=\boldsymbol{m}(j)$ for $(1-\epsilon)n$ values of $j\in[n]$. 

Therefore, we have shown so far that with transmission of $R=n(I(\boldsymbol{M};C|K)+\epsilon')$ bits, we can ensure that with high probability, $\boldsymbol{M}^n$ matches $\widehat{\boldsymbol{M}}^n$ on $(1-\epsilon)$ fraction of its entries. However, we need the whole $\boldsymbol{M}^n$ to be equal to $\widehat{\boldsymbol{M}}^n$ with high probability. We resolve this below, but observe that since the length of the messages are $H(M_i^n)=nH(M_i)$, we have indeed reached the index code rate
$$\left(\frac{H(M_1)}{I(\boldsymbol{M};C|K)+\epsilon'}, \frac{H(M_2)}{I(\boldsymbol{M};C|K)+\epsilon'}, \cdots, \frac{H(M_t)}{I(\boldsymbol{M};C|K)+\epsilon'}\right).$$ 
Let us go back to the fact that with high probability $1-\delta$, we have that $\boldsymbol{M}^n$ matches $\widehat{\boldsymbol{M}}^n$ on $(1-\epsilon)$ fraction of its entries, and not entirely. We show that this can be fixed with a negligible decrease in index coding rates. The idea is that by Fano's inequality $$\frac 1n H(\boldsymbol{M}^n|\widehat{\boldsymbol{M}}^n)\leq \frac 1n+\delta H(\boldsymbol{M})+(1-\delta)\epsilon H(\boldsymbol{M})$$
can be made as close as we want to zero. Thus, using Slepian-Wolf theorem, conveying $\boldsymbol{M}$ with side information $\widehat{\boldsymbol{M}}^n$ at the decoder will require negligible amount of communication. To achieve this, one has to take $N$ i.i.d.\ repetitions of $\boldsymbol{M}^n$ and $\widehat{\boldsymbol{M}}^n$, and then use the  Slepian-Wolf theorem to ensure that repetitions of $\boldsymbol{M}^n$ are recovered with high probability.


\section{Conclusion}
\label{Conclusion_Sec}

In this paper, we studied the index coding problem in the presence of an eavesdropper. Assuming that a common as well as a set of  dedicated private keys are shared among the transmitter and legitimate receivers, we obtained a condition on keys' entropies by which the index code could be transmitted securely. In Theorem \ref{Main_theorem}, we made a relationship between the secure index coding problem to one without secrecy, and showed that the generalized one-time pad strategy is optimal up to a multiplicative constant for the secure index coding problem. In other words, we showed that the conventional index coding rate region determines the cone of the secure rate region, which is equal to the cone of the generalized one-time pad strategy. Theorem \ref{Main_theorem_Linear} presents a similar statement to the Theorem \ref{Main_theorem} for the linear case. Moreover, we showed in Theorem \ref{Lemma_Equivalent_epsilon_zero} that relaxing the secrecy condition from perfect to weak secrecy does not change the rate region when we have an $\epsilon$-error decoding condition. As a future work, one can study the effect of adversary's side information and/or capability of corrupting the public communication.


\bibliographystyle{IEEEtran}

\end{document}